\documentclass[conference,twocolumn, twoside]{IEEEtran}
\usepackage{amssymb,epsfig,graphics, graphicx, url,times,mathrsfs,algorithm,algorithmic,amsmath,array,latexsym,fancyhdr,xspace,wrapfig, xcolor,multirow,amsthm,caption,cite,helvet,sidecap,bm,comment,upref,bbm}
\usepackage[normalem]{ulem}
\usepackage{balance}
\usepackage{xargs}

\usepackage[normalem]{ulem}

\newcommand{\rvY}{\texttt{Y}}
\definecolor{sacramentostategreen}{rgb}{0.0, 0.34, 0.25}

\usepackage[inline]{enumitem}
\usepackage{arydshln}
\usepackage[makeroom]{cancel}
\usepackage{varwidth}
\usepackage{subcaption}
\usepackage{pgfplots}
\usepackage{nicefrac}
\usepackage{adjustbox}

\usepackage[capitalize]{cleveref}

\usepackage{tikz}
\usetikzlibrary{shapes.geometric}
\usetikzlibrary{shapes,snakes,patterns}
\usetikzlibrary{arrows,positioning,automata,calc}
\usetikzlibrary{plotmarks}
\usepackage{ctable}

\definecolor{green1}{rgb}{0,0.5,0}
\definecolor{magenta}{rgb}{1.0, 0.11, 0.81}
\definecolor{mulberry}{rgb}{0.77, 0.29, 0.55}
\definecolor{xgray}{rgb}{0.9, 0.9, 0.9}
\def \blue{\color{blue}}

\def \bes{\begin{equation*}}
\def \ees{\end{equation*}}
\def \bas{\begin{align*}}
\def \eas{\end{align*}}
\def \be{\begin{equation}}
\def \ee{\end{equation}}
\def \bbm{\begin{bmatrix}}
\def \ebm{\end{bmatrix}}

\def \rvA{\texttt{A}}
\def \rvB{\texttt{B}}

\def \rvR {\texttt{R}}

\def \rvX{\texttt{X}}
\def \rvY{\texttt{Y}}

\def \F{\mathbb{F}}

\newcommand{\Aij}{\rvA_{\{i,j\}}}
\newcommand{\Rij}{\rvR_{\{i,j\}}}
\newcommand{\Xij}[2]{\rvX_{\{#1,#2\}}}

\newcommand{\bfx}{{\boldsymbol x}}

\newcommand{\bfA}{{\mathbf A}}
\newcommand{\bfB}{{\mathbf B}}

\newcommand{\bfR}{{\mathbf R}}

\newcommand{\bfX}{{\mathbf X}}

\newcommand{\bfXij}{\bfX_{\{i,j\}}}
\newcommand{\I}{\textrm{I}}
\newcommand{\mutinf}{\I_q}
\newcommand{\kl}[2]{\text{D}_{\text{KL}} \left(#1 \Vert #2 \right)}
\newcommand{\entropy}{\textrm{H}_q}
\newcommand{\leakage}[1]{\textrm{L}_#1}

\newcommand{\sr}{\ensuremath s_{\bfR}}
\newcommand{\sar}{\ensuremath s_{\bfA+\bfR}}
\newcommand{\savg}{\ensuremath s_{\text{avg}}}
\newcommand{\sdiff}{\ensuremath s_{\delta}}

\newcommand{\sparsity}{\textrm{S}}
\newcommand{\pz}{p_{1}}
\newcommand{\pzinv}{p_{1}^{\text{inv}}}
\newcommand{\pnz}{p_{3}}
\newcommand{\pext}{p_{2}}
\newcommand{\pnzinv}{p_{2,3}^\text{inv}}

\newcommand{\worker}{node}
\newcommand{\Worker}{Node}
\newcommand{\workers}{nodes}

\newtheorem{theorem}{Theorem}

\newtheorem{lemma}{Lemma}

\newtheorem{definition}{Definition}

\newtheorem{remark}{Remark}

\usetikzlibrary{decorations.pathreplacing,calc}
\tikzset{brace/.style={decorate, decoration={brace}},
 brace mirrored/.style={decorate, decoration={brace,mirror}},
}

\newcounter{brace}
\newcommand{\Fq}{\mathbb{F}_q}
\newcommand{\Fqstar}{\Fq^\star} 
\newcommand{\pa}[1][a]{\mathrm{P}_{\rvA}}
\newcommand{\paandapr}{\mathrm{P}_{\rvA,\rvA+\rvR}}
\newcommand{\papr}{\mathrm{P}_{\rvA+\rvR}}
\newcommand{\paandr}[1][a]{\mathrm{P}_{\rvA,\rvR}}
\newcommand{\pr}[1][r]{\mathrm{P}_{\rvR}}
\newcommand{\optimizer}{\min\limits_{\mathcal{P}}} %
\newcommandx{\pra}[2][1=r, 2=a]{p_{#1,#2}}
\newcommand{\lossopt}{\mathrm{L}_\text{opt}}
\newcommand{\totalleakage}{\mathrm{L}(\pz,\pzinv,\pext,\pnz,\pnzinv)}

\newcommand{\loss}{\mathcal{L}(\pz,\pzinv,\pext,\pnz,\pnzinv,\lambda_1,\lambda_2,\lambda_3,\lambda_4)}
\newcommand{\lossabbrev}{\mathcal{L}}
\newcommand{\grad}{\nabla_{\pz,\pzinv,\pnz,\pext,\pnzinv,\lambda_1,\lambda_2,\lambda_3,\lambda_4}}
\newcommand{\srinv}{\sr^\text{inv}}
\newcommand{\sarinv}{\sar^\text{inv}}
\newcommand{\qfac}{\theta}

\newcommand{\setsubsharesar}[1]{\ensuremath{\mathcal{W}_#1^{\bfA+\bfR}}}
\newcommand{\setsubsharesr}[1]{\ensuremath{\mathcal{W}_#1^{\bfR}}}

\newcommand{\m}{\ensuremath{r}}
\newcommand{\n}{\ensuremath{\ell}}
\newcommand{\nbworkers}{\ensuremath{n}}
\newcommand{\str}{\ensuremath{\xi}}
\newcommand{\transpose}{\mathsf{T}}

\newcommand\scalemath[2]{\scalebox{#1}{\mbox{\ensuremath{\displaystyle #2}}}}

\IEEEoverridecommandlockouts

\begin{document}

\title{{Efficient Private Storage of Sparse Machine Learning Data}}
\vspace{-0.4cm}
\author{\IEEEauthorblockN{
Marvin Xhemrishi, Maximilian Egger and Rawad Bitar}\\ 
\vspace{-0.48cm}
\IEEEauthorblockA{Institute for Communications Engineering, Technical University of Munich, Munich, Germany\\
\{\texttt{marvin.xhemrishi, maximilian.egger, rawad.bitar}\}\texttt{@tum.de}}
\thanks{Thiw work was partially funded by the DFG (German Research Foundation) project under Grant Agreement No. WA 3907/7-1 {and by the Technical University of Munich - Institute for Advanced Studies, funded by the German Excellence Initiative
and European Union Seventh Framework Programme under Grant Agreement
No. 291763.}
}
\vspace{-1cm}
}
\vspace{-1cm}

\maketitle
\begin{abstract}
 We consider the problem of maintaining sparsity in private distributed storage of confidential machine learning data. In many applications, e.g., face recognition, the data used in machine learning algorithms is represented by sparse matrices which can be stored and processed efficiently. However, mechanisms maintaining perfect information-theoretic privacy require encoding the sparse matrices into randomized dense matrices. It has been shown that, under some restrictions on the storage nodes, sparsity can be maintained at the expense of relaxing the perfect information-theoretic privacy requirement, i.e., allowing some information leakage. In this work, we lift the restrictions imposed on the storage nodes and show that there exists a trade-off between sparsity and the achievable privacy guarantees. We focus on the setting of non-colluding nodes and construct a coding scheme that encodes the sparse input matrices into matrices with the desired sparsity level while limiting the information leakage.
\end{abstract}

\section{Introduction}\label{sec:Intro}
In the era of big data, efficient and private distributed data storage became a necessity. Analyzing tremendous amount of information is now possible thanks to machine learning algorithms that take data as input and output a compressed model that represents the data. A new input data can then be compared to the stored data by processing it using only the model. In several machine learning applications, such as face recognition, the input data is represented by sparse matrices (having a large number of zero entries) \cite{wright2008robust,madarkar2021sparse}. Furthermore, the model representing a neural network is also known to be represented and stored in the form of a sparse matrix.

We focus on efficient and private distributed storage of sparse data matrices. Sparse matrices are known to be efficiently stored and operated on. The privacy measure we use is information-theoretic privacy as it assumes no bound on the computational power of the potential eavesdroppers. The techniques used in the literature to maintain information-theoretic privacy of matrices destroy the structure of the matrix, i.e., result in dense matrices. On a high level, maintaining information-theoretic privacy of a matrix $\bfA$ can be ensured by generating a random matrix $\bfR$ independently from $\bfA$ and uniformly at random from the same alphabet to which $\bfA$ belongs. Then, storing $\bfR$ and $\bfA+\bfR$ (where the addition is element-wise), ensures the privacy of $\bfA$. Observe here that $\bfR$ and $\bfA+\bfR$ are statistically independent from $\bfA$. However, they are both dense matrices. Hence, perfect privacy of $\bfA$ is ensured, but its sparsity cannot be leveraged to store $\bfR$ and $\bfA+\bfR$ efficiently.

Our goal is to design schemes that allow a private distributed storage of $\bfA$ while leveraging its sparsity. Going back to our illustrative example, this is made possible by designing the random matrix $\bfR$ dependently on $\bfA$ so that both $\bfR$ and $\bfA+\bfR$ are sparse matrices and reveal as little information about $\bfA$ as possible. The matrices $\bfR$ and $\bfA+\bfR$ are called shares.

Techniques ensuring information-theoretic privacy in distributed storage date back to the seminal works of Shamir and Blakely~\cite{S79,Blakley1979} which can be seen as a generalization of Shannon's one-time pad~\cite{shannon_one_time_pad}. Those techniques are known as threshold secret sharing schemes. In~\cite{McESa81}, McElice and Sarwate generalized threshold secret sharing further by connecting them to Reed-Solomon codes. An $(n,k,z)$ secret sharing encodes $k-z$ files into $n$ shares, stored on $n$ different nodes, such that any $z$ shares reveal no information about the files and any $k$ shares suffice to reconstruct the stored files.

To cope with the increase in the amount of data stored and handled, efficient secret sharing schemes are being investigated. For example, secret sharing schemes with communication efficiency are studied in~\cite{4418504,HLKBtrans,BRIT17,martinez2018communication}. The goal is to reduce the amount of data communicated to reconstruct the secret at the expense of increasing the amount of accessed nodes. Other efficient modern secret sharing schemes are secure (against malicious attacks) and repairable secret sharing \cite{shah2011information,pawar2011securing,8469091}, secure communication-efficient secret sharing schemes \cite{bitar2020communication} and communication-efficient secret sharing schemes with small share size (data storage) \cite{9594824}. Efficient security is guaranteed by leveraging the weakness of the malicious nodes and transforming errors into erasures to reduce the required redundancy. We refer the interested reader to \cite{pawar2011securing,bitar2020communication} for more details. However, all the  aforementioned coding schemes destroy the structure of the input data and generate shares (data to be stored) that look completely random, i.e., dense shares.

Motivated by the efficient use of sparse matrices, we consider secret sharing schemes that encode sparse data matrices into sparse shares. Hence, allowing for efficient storage of the shares since sparse shares can be compressed and are also efficiently processed in case they are used in computations \cite{adaptive_sparse_matrix,implementing_sparse_matrix_vector}. Sparsity in non-private distributed linear computations with tolerance of slow or unresponsive compute nodes, called stragglers, is shown to speed up the computation process. The authors of~\cite{coded_sparse_mm} considered the problems of encoding two sparse matrices into smaller sparse chunks sent to external nodes for computations. The proposed scheme relies on using LT codes~\cite{LT_codes} with a sparse encoding principle, i.e., using a Soliton distribution with high occurrence of low degrees. In~\cite{coded_sparse_matrix_leverages_partial_stragglers}, the authors propose to divide the matrices into smaller ones, without any encoding and use fractional repetition codes \cite{el2010fractional} to distribute the small matrices to the computing nodes. For an extra layer of protection against stragglers, each computing node could also receive an additional encoded sub-matrix.

The problem of designing secret sharing schemes with sparse shares has been recently initiated in \cite{sparse_ISIT} for distributed private matrix-vector multiplication. The multiplication being linear with no privacy requirement on the vector, the problem of private distributed storage and private distributed computing become equivalent. In the terminology of this paper, the authors of \cite{sparse_ISIT} construct a secret sharing scheme that trades privacy guarantees for the sparsity of the shares, since insisting on perfect information-theoretic privacy does not allow any sparsity in the shares. The setting considered in \cite{sparse_ISIT} consists of two non-communicating clusters of nodes. A special type of shares is created for each cluster, hence separately controlling the sparsity and the amount of information leaked from each cluster of nodes. {For a pure storage scenario, it can be beneficial to first compress the sparse private data and then perform classical secret sharing. However, we consider the general case where the stored data is later used in form of linear computations, similar to the setting in \cite{sparse_ISIT}.}%

In this work, we remove the constraint of non-communicating clusters of nodes. For a fixed desired sparsity of the shares, we provide a convex minimization problem that outputs the minimum leakage possible. Under some assumptions, we solve the problem and construct secret sharing schemes that achieve the minimum leakage of information.

The paper is organized as follows. In Section~\ref{sec:Preliminaries}, we set the notation and formalize the considered problem. We show the existence of a trade-off between sparsity and privacy in Section~\ref{sec:trade-off}. For a special parameter regime, we give in Section~\ref{sec:scheme} a scheme that optimizes the leakage for a desired sparsity guarantee. Section~\ref{sec:conc} concludes the paper.

\section{Preliminaries}\label{sec:Preliminaries}
\emph{Notation:} Matrices and vectors are denoted by uppercase and lowercase bold letters, e.g., $\bfX$ and $\bfx$, respectively. The $(i,j)$-th entry of a matrix $\bfX$ is denoted by $\bfXij$. We use uppercase \emph{typewriter} letters for random variables, e.g., $\rvY$. The random variables representing a matrix $\bfX$ and its $(i,j)$-th entry $\bfXij$ are denoted respectively by $\rvX$ and $\Xij{i}{j}$. A finite field of cardinality $q$ is denoted by $\F_q$ and we define $\F_q^* \triangleq \F_q\setminus \{0\}$. Calligraphic letters denote sets, e.g., $\mathcal{X}$. For a positive integer $b$, we define $[b]\triangleq \{1,2,\dots, b\}$. Given $b$ random variables $\rvY_1,\dots,\rvY_b$ and a set $\mathcal{I}\subseteq [b]$, we denote by $\{\rvY_i\}_{i\in\mathcal{I}}$ the set of random variables indexed by $\mathcal{I}$, i.e., $\{\rvY_i\}_{i\in\mathcal{I}} \triangleq \{\rvY_{i} | i\in \mathcal{I}\}$. For a random variable $\rvX \in \mathbb{F}_q$ we denote its probability mass function (PMF) over $\mathbb{F}_q$ by $\textrm{P}_\rvX = \left[p_1, p_2, \dots, p_q \right]$, i.e., $\Pr(\rvX = i) = p_i$ for all $i\in \mathbb{F}_q$, and its $q$-ary entropy by $\textrm{H}_q(\rvX)$ and $\textrm{H}_q(\left[p_1, p_2, \dots, p_q \right])$ interchangeably. Given two random variables $\rvX$ and $\rvY$, their $q$-ary mutual information is denoted by $\I_q(\rvX; \rvY)$. {The Kullback-Leibler (KL) divergence between the PMFs of random variables $\rvX$ and $\rvY$ reads as $\kl{\textrm{P}_{\rvX}}{\textrm{P}_{\rvY}}$.}
We follow \cite{sparse_ISIT} and define the sparsity of a matrix as follows.

{\definition(Sparsity level of a matrix) {The sparsity level $\sparsity(\rvX)$ of a matrix $\rvX$ with entries independently and identically distributed is equal to the probability of the $(i,j)$-th entry $\rvX_{\{i,j\}}$ of $\rvX$ being equal to $0$, i.e., $$ \sparsity(\rvX) = \Pr\{\Xij{i}{j} = 0\}.$$}}

\emph{System model:} We study secret sharing schemes that take a sparse matrix $\bfA$ as input and output two sparse shares $\bfR$ and $\bfA+\bfR$ each in $\F_q^{\m\times \n}$. Dividing the shares row-wise and carefully applying a fractional repetition code gives a secret sharing scheme with larger parameters, $n$ and $k$. The entries of the input matrix $\bfA$ are assumed to be independently and identically distributed, i.i.d. Though restrictive, this distribution serves as a starting point towards understanding the trade-off between sparsity and privacy of secret sharing schemes. In particular, we assume that $\Pr\{\Aij = 0\} = s$ and $\Pr\{\Aij = a\} = \frac{1-s}{q-1}$ for all $a \in \F_q^*$. Hence, $\bfA$ is considered to have a sparsity level $\sparsity(\rvA) = s > q^{-1}$ since the special case where $s\leq q^{-1}$ can be satisfied using classical secret sharing schemes.  We focus on information-theoretic privacy. Given two random variables $\rvA$ and $\rvB$, we say that observing a realization $\bfB$ of $\rvB$ leaks $\varepsilon \triangleq \mutinf(\rvA;\rvB)$ information about $\rvA$. If the leakage $\varepsilon$ is zero, we say that perfect privacy is attained.

A classical secret sharing scheme generates a random matrix $\rvR$, independently from $\rvA$ and uniformly at random from the same alphabet from which $\rvA$ is drawn. The scheme then outputs two shares that are realizations of the random variable $\rvR$ and the random variable $\rvA+\rvR$ such that $\mutinf(\rvA;\rvR) = \mutinf(\rvA;\rvA+\rvR) = 0$. The sparsity of $\rvR$ and $\rvA+\rvR$ is equal to $q^{-1}$. A \emph{sparse secret sharing scheme} outputs two shares that are realizations of two random variables $\rvR$ and $\rvA+\rvR$ such that $\mutinf(\rvA;\rvR) = \varepsilon_1 \geq 0$ and $\mutinf(\rvA;\rvA+\rvR) = \varepsilon_2\geq 0$. The main difference is that the random variable $\rvR$ is drawn \emph{dependently} from $\rvA$. Hence, the sparsity of $\rvR$ and $\rvA+\rvR$ can be made larger than $q^{-1}$ at the expense of choosing $\varepsilon_1> 0$ and/or $\varepsilon_2>0$.

We follow the nomenclature of \cite{sparse_ISIT} and call the share $\bfR$ as the \emph{padding matrix} and the share $\bfA+\bfR$ as the \emph{padded matrix}. We define $\sr \triangleq \sparsity(\rvR)$ and $\sar \triangleq \sparsity(\rvA+\rvR)$ to be the sparsity levels of the shares and $\leakage{1} \triangleq \mutinf(\rvA;\rvR)$ and $\leakage{2} \triangleq \mutinf(\rvA;\rvA+\rvR)$ to be the information leakage from each share. For given sparsity levels $\sr,\sar> q^{-1}$, we provide a convex minimization problem whose solution gives the smallest possible leakage. For a special regime, we solve the problem and provide schemes that attain the smallest leakage for the desired sparsity. We observe that the leakage is minimized when $\sr$ is chosen to be equal to $\sar$ and that for this case, the minimum leakage increases with the desired value of $\sr$.

\begin{remark}
In practice, the positive leakage can be understood as leaking some side information about the private data. For example, if the private matrix represents a face, leaking a small amount of information could be understood as leaking some details about the face, such as the size or the contour of the face, but not leaking the whole face.
\end{remark}

\section{Minimizing Leakage for Fixed Sparsity} \label{sec:trade-off}
The main idea behind creating sparse shares in a secret sharing scheme is to design the random matrix $\rvR$ dependently from the input matrix $\rvA$. If $\rvR$ is generated uniformly at random from $\F_q$, then perfect privacy is achieved, but no sparsity in $\rvR$ and $\rvA+\rvR$ is maintained. 

For desired sparsity levels $\sr$ and $\sar$, the problem then translates to designing a probability distribution of $\rvR$ that minimizes the leakage. Since the entries $\rvA_{\{i,j\}}$ of $\rvA$ are assumed to be independent, we treat the entries of $\rvR$ independently. Hence, we need to find values $p_{r,a} \triangleq \Pr(\rvR_{\{i,j\}} = r|\rvA_{\{i,j\}} = a)$ for all $r,a \in \mathbb{F}_q$ that form a probability distribution. The sparsity constraints then restrict the $p_{r,a}$'s to ensure that $\Pr(\rvR_{\{i,j\}} = 0) = \sr$ and $\Pr(\rvR_{\{i,j\}} + \rvA_{\{i,j\}} = 0) = \sar$.

The choice of the $\pra$'s affects the leakage $\leakage{1}$ and $\leakage{2}$. To minimize the leakage, we analyze the total leakage resulting from the sum $\leakage{1}+\leakage{2}$. Our interest in minimizing the sum of the leakages stems from the need to use a fractional repetition code which may assign shares of the form $\rvR$ and $\rvA+\rvR$ to a certain node. By writing $\leakage{1}+\leakage{2}$ using the definition of mutual information, minimizing the total leakage subject to the sparsity constraints results in the following minimization problem. Let $\mathcal{P} \triangleq \{p_{r,a}: r,a\in \mathbb{F}_q\}$, then
\begin{align*}
    \lossopt &= \optimizer \leakage{1} + \leakage{2} = \optimizer \mutinf\left(\rvR; \rvA\right) + \mutinf\left(\rvA+\rvR; \rvA\right) \\
    &=\optimizer \begin{aligned}[t] &\kl{\paandr}{\pa \pr} + \kl{\paandapr}{\pa \papr} \end{aligned} \\
    &=\optimizer \!\! \sum_{a,b\in\mathbf{F}_q} \!\! \pa(a) \left( \pra[b-a][a] \log\frac{\pra[b-a][a]}{\papr(b)} + \pra[b][a] \log\frac{\pra[b][a]}{\pr(b)} \right)\!.
\end{align*}
subject to
\begin{align*}
    \forall a\in\Fq: \pra[0][a] + \sum_{r\in\Fqstar} \pra[r][a] - 1 &= 0 \\[-5pt]
    \pra[0][0] \cdot \pa(0) + \sum_{a\in\Fqstar} \pra[0][a] \cdot \pa(a) -\sr &= 0 \\[-5pt]
    \pra[0][0] \cdot \pa(0) + \sum_{a\in\Fqstar} \pra[-a][a] \cdot \pa(a) -\sar &= 0
\end{align*}
This non-linear convex optimization problem with multiple linear constraints can be solved using the method of Lagrange multipliers \cite{rockafellar1993lagrange}. For analytical tractability and for ease of analysis, we restrict our attention to PMFs of $\rvR$ of the following form.
\begin{align}
 \label{eq:dependent_on_0}
    \Pr\{\Rij = r \lvert \Aij = 0\} \!&\!= \begin{cases} 
    \pz, &r = 0 \\   \pzinv , &r \neq 0,
    \end{cases}\\
\label{eq:dependent_on_nz}
    \Pr\{\Rij = r \lvert \Aij = a\} \!&\!= \begin{cases} 
    \pext, &r = 0 \\ \pnz, &r = -a \\ \pnzinv , &r \not\in \{0,-a\},
    \end{cases}
\end{align}

where $r\in \F_q$, $a \in \F_q^*$, and $-a$ is the additive inverse of $a$ in $\F_q^*$. Further, $\pz,\pext,\pnz, \pzinv \triangleq (1-\pz)/(q-1)$ and $\pnzinv \triangleq (1-\pext-\pnz)/(q-2)$ are non-negative numbers smaller than $1$. This distribution puts an emphasis on three ideas: inheriting in $\rvR$ and in $\rvA+\rvR$ a zero from $\rvA$ (with probability $\pz$); creating a zero in $\rvR$ at the expense of revealing the $(i,j)$-th entry of $\rvA$ in $\rvA+\rvR$ (with probability $\pext$); and creating a zero\footnote{Note that compared to \cite{sparse_ISIT}, this construction is symmetric, i.e., it allows creating zeros in both $\rvR$ and $\rvA+\rvR$.} in $\rvA+\rvR$ at the expense of revealing the additive inverse of the $(i,j)$-th entry of $\rvA$ in $\rvR$ (with probability $\pnz$).

The sparsity of $\rvR$ and $\rvA+\rvR$ and their leakage about $\rvA$ when generated as shown above are stated in \cref{lemma:sparsity,lemma:leakage}, respectively.

\begin{lemma}\label{lemma:sparsity}
Following the rules in \cref{eq:dependent_on_0,eq:dependent_on_nz} to construct two matrices $\bfR$ and $\bfA+\bfR$ dependently on $\rvA$ with sparsity $s$, results in the following sparsity levels
\begin{align*}
    \sr &= \pz s + \pext (1-s), \\
    \sar &= \pz s + \pnz(1-s).
\end{align*}
\end{lemma}

The proof is straightforward and is omitted for brevity. Note that, by construction, for $0\leq \pext,\pnz\leq 1$, the difference $\vert \sar - \sr \vert$ can be at most $(1-s)$.

\begin{lemma} \label{lemma:leakage}
Considering PMFs of the form given in \cref{eq:dependent_on_0,eq:dependent_on_nz}, the total leakage $\totalleakage \triangleq \mathrm{L}_1+\mathrm{L}_2$ is given by
\begin{align*}
&\totalleakage = s \bigg[ \pz \bigg( \log \frac{\pz}{\sar} + \log \frac{\pz}{\sr} \bigg) \nonumber\\
&+(q-1)\pzinv \bigg( \log \frac{\pzinv}{\sarinv} + \log \frac{\pzinv}{\srinv}\bigg) \bigg] + (1-s) \nonumber\\
&\cdot \bigg[ \pext \bigg( \log \frac{\pext}{\sar} + \log \frac{\pext}{\srinv} \bigg) + \pnz \bigg( \log \frac{\pnz}{\sarinv} + \log \frac{\pnz}{\sr} \bigg) \nonumber\\ &+(q-2) \pnzinv \bigg( \log \frac{\pnzinv}{\sarinv} + \log \frac{\pnzinv}{\srinv} \bigg)
\bigg],
\end{align*}
where $\srinv \triangleq (1-\sr)/(q-1)$, $\sarinv \triangleq (1-\sar)/(q-1)$.
\begin{proof}
This statement follows from expressing $\leakage{1}+\leakage{2}$ in terms of the KL-divergence and simplifying for distributions of the form given in \cref{eq:dependent_on_0,eq:dependent_on_nz}. %
\end{proof}
\end{lemma}

Furthermore, by considering PMFs of the form given in \cref{eq:dependent_on_0,eq:dependent_on_nz}, the constraints simplify to
\begin{align}
    c_1(\pz,\pzinv) &\triangleq \pz + (q-1) \pzinv - 1 &&= 0, \label{eq:constraint1} \\
    c_2(\pext,\pnz,\pnzinv) &\triangleq \pext + \pnz + (q-2) \pnzinv - 1 &&= 0, \label{eq:constraint2} \\
    c_3(\pz,\pnz) &\triangleq \pz s + \pnz (1-s) - \sr &&= 0, \label{eq:constraint3} \\
    c_4(\pz,\pext) &\triangleq \pz s + \pext (1-s) - \sar &&= 0. \label{eq:constraint4}
\end{align}

For clarity of presentation, we define the following quantities that will be helpful in stating the result of \cref{thm:optimal_pmf}.
\begin{align*}
\qfac & \triangleq (q-2)^2/(q-1), \alpha \triangleq s^2 (4+\qfac), \delta \triangleq -\qfac \sar \sr,\\
\beta & \triangleq 4s(1-s-\sar-\sr)-\qfac s (\sr+\sar+s),\\
\gamma & \triangleq (1-s-\sar-\sr)^2+\qfac(s\sar + s\sr + \sar \sr).
\end{align*}

\begin{theorem} \label{thm:optimal_pmf}
Given desired sparsities $\sar$ and $\sr$, the optimal PMF of the form given in \cref{eq:dependent_on_0,eq:dependent_on_nz} that minimizes the leakage is obtained by setting $\pz=\pz^\star$ which is the root of the polynomial $\alpha \pz^3 + \beta \pz^2 + \gamma \pz + \delta$ that satisfies $\max\{\sr+\sar-1+s,0\} \leq 2 \pz s \leq 2\min\{s,\sr,\sar\}$. Having $\pz$, the remaining unknowns $\pzinv,\pext,\pnz$ and $\pnzinv$ directly follow from \eqref{eq:constraint1}-\eqref{eq:constraint4}.
\begin{proof}
To find the distribution of $\rvR$ that minimizes the leakage, we utilize the method of Lagrange multipliers to combine the objective function, i.e., the total leakage from \cref{lemma:leakage} with the constraints in \eqref{eq:constraint1}-\eqref{eq:constraint4}. Thus, the objective function to be minimized can be expressed as
\begin{align*}
    \lossabbrev &\triangleq \loss \\
    & = \totalleakage + \lambda_1 c_1(\pz,\pzinv) + \\
    & + \lambda_2 c_2(\pext,\pnz,\pnzinv) + \lambda_3 c_3(\pz,\pnz) + \lambda_4 c_4(\pz,\pext).
\end{align*}

This objective can be minimized by setting the gradient to zero, i.e., $\grad \lossabbrev = 0$, which amounts to solving a system of nine equations with nine unknowns. Simplifying $\nabla_{\pz,\pzinv,\pext,\pnz,\pnzinv} \lossabbrev = 0$ yields
\begin{align}
    \pz (\pnzinv)^2 &= \pzinv \pext \pnz. \label{eq:grad_obj_compact}
\end{align}

Solving \eqref{eq:constraint1}-\eqref{eq:constraint4} for $\pzinv,\pext,\pnz,\pnzinv$ as a function of $\pz$ (i.e., each a polynomial of degree one), and inserting the result into \eqref{eq:grad_obj_compact} amounts to finding the root of a polynomial of degree three in $\pz$ that satisfies the constraints, as stated in the theorem. Since the problem is convex, the parameters yielding the local minimum must be a global minimizer to the constrained optimization problem.
\end{proof}
\end{theorem}

Our numerical results show that when the entries of $\rvA$ are i.i.d. with the non-zero entries being uniformly distributed over $\Fq^*$, optimizing the leakage over the PMFs as in \cref{eq:dependent_on_0,eq:dependent_on_nz} yields the same result as optimizing the leakage over all PMFs with $q^2$ unknowns $\pra$. This justifies the usage of PMFs as in \cref{eq:dependent_on_0,eq:dependent_on_nz}.

When dividing the shares into multiple matrices called sub-shares and applying fractional repetition codes to distribute the smaller sub-shares to storage nodes, every \worker\ stores multiple sub-shares as detailed in \cref{sec:scheme}. To this end, we focus on a desired average sparsity $\savg \triangleq \frac{\sr+\sar}{2}$ and analyze the minimum possible leakage for varying differences of the sparsities\footnote{As the construction is symmetric, we choose $\sar \geq \sr$ without loss of generality.} $\sdiff \triangleq \sar - \sr$. We plot in \cref{fig:converse} the optimal leakage according to \cref{lemma:leakage} obtained by the PMF in \cref{thm:optimal_pmf} as a function of the average sparsity $\savg$ for two different values of $\sdiff$.

\begin{remark}
We observe that the minimum possible total leakage is obtained for $\sar = \sr$, i.e., according to \cref{lemma:sparsity} if $\pext=\pnz$, and consequently when the leakages $\leakage{1}$ and $\leakage{2}$ are equal.
If one cares about the maximum leakage, i.e., $\max\{\leakage{1}, \leakage{2}\}$, then our observation implies that the parameters for $\sr=\sar$ that minimize the total leakage also minimize the maximum leakage.
\end{remark}

\begin{figure}[t]
    \centering
    \resizebox{.48\textwidth}{!}{\input{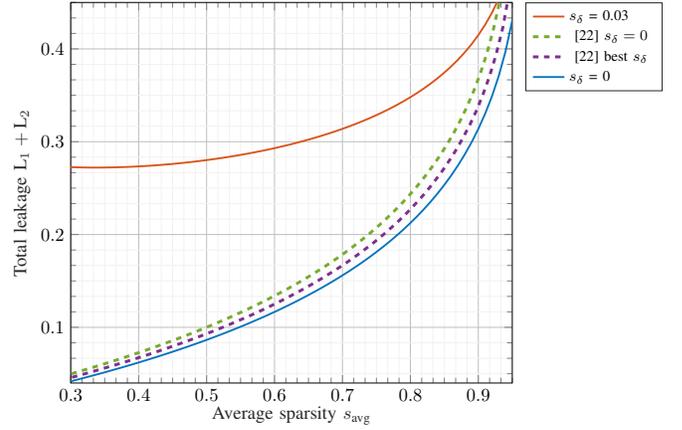}}
    \caption{We provide the optimal element-wise total leakage $(\leakage{1}+\leakage{2})/ \m \n$ over different values of $\savg$ with varying $\sdiff$ for $q=256$ and $s=0.95$ (solid lines). We compare our bounds with the element-wise total leakage of the scheme in~\cite{sparse_ISIT} for its best $\sdiff$ and for $\sdiff =0$ (dashed lines).}
    \label{fig:converse}
\end{figure}

\section{Sparse Secret Sharing Schemes}\label{sec:scheme}
We show how to divide $\rvR$ and $\rvA+\rvR$ created as explained previously and combine them with a fractional repetition code to obtain a sparse secret sharing scheme for a desired even number of storage nodes $\nbworkers$ and a desired parameter $k> \frac{N}{2}$. Extra care should be taken in distributing the sub-shares to avoid leaking extra information about $\rvA$.

\subsection{Creation and Distribution of the Sub-Shares}\label{subsec:colluding}

To create the sub-shares, the share $\bfR$ is divided into $\nbworkers$ sub-shares $\bfR_0,\dots, \bfR_{\nbworkers-1}$, i.e., $\bfR = \left[\bfR_0^\transpose, \bfR_1^\transpose,\dots,\bfR_{\nbworkers-1}^\transpose\right]^\transpose$. Similarly, the share $\bfA + \bfR$ is divided row-wise into $\nbworkers$ sub-shares $(\bfA+\bfR)_0, \dots, (\bfA+\bfR)_{\nbworkers-1}$. To create an $(n,k,z)=(\nbworkers,\nbworkers,1)$ sparse secret scheme, we distribute the sub-shares as follows. For all $i\in \{0,\dots,\nbworkers-1\}$, \worker\ $i$ stores the sub-shares $(\bfA+\bfR)_i$ and $\bfR_{\nbworkers/2 + i \mod \nbworkers}$. Note that the distribution of the sub-shares to the nodes can be arbitrary as long as each node obtains two sub-shares and no node obtains $\bfR_i$ and $(\bfA+\bfR)_i$. However, we chose this particular task distribution since it can be easily expanded to decrease the value of $k$ at the expense of increasing the number of sub-shares given to each storage node. Having schemes with $k<\nbworkers$ allows tolerance of slow or unresponsive nodes referred to as stragglers.

We use a fractional repetition codes. To tolerate $\str$ stragglers, i.e., set $k=\nbworkers-\str$, each node $i\in \{0,\dots,\nbworkers-1\}$ obtains the following set of sub-shares $\setsubsharesar{i} \triangleq \{(\bfA+\bfR)_{i+j\mod \nbworkers}\}_{j=0}^\str$ and $\setsubsharesr{i} \triangleq \{\bfR_{\nbworkers/2+i+ j\mod \nbworkers}\}_{j=0}^\str$. Observe that for $\str \in \{0,\dots,\frac{\nbworkers}{2}-1\}$ this distribution does not give sub-shares $\bfR_i$ and $(\bfA+\bfR)_i$ to any \worker. Hence, for $\str \in \{0,\dots,\frac{\nbworkers}{2}-1\}$ the proposed share distribution results in an $(n,k,z) = (\nbworkers, \nbworkers-\str, 1)$ sparse secret sharing scheme. An illustration of scheme for $\nbworkers=2$ and $\str = 1$ is illustrated in \cref{fig:fractional_repetition_sparse}.

\subsection{Analysis of the proposed scheme}\label{subsec:privacy_scheme}
We analyze the properties of the proposed sparse secret sharing scheme in terms of the information leakage, straggler tolerance and storage cost.

\begin{theorem}\label{thm:scheme_analysis}
Let $\bfR \in \Fq^{\m\times \n}$ and $\bfA+\bfR \in \Fq^{\m\times \n}$ be two shares encoded as shown in \cref{sec:trade-off} with average sparsity $\savg$. Given an even number of storage nodes $\nbworkers$ and an integer $\str \in \{0,\dots,\frac{\nbworkers}{2}-1\}$, storing the sub-shares $\setsubsharesar{i} \triangleq \{(\bfA+\bfR)_{i+j\mod \nbworkers}\}_{j=0}^\str$ and $\setsubsharesr{i} \triangleq \{\bfR_{\nbworkers/2+i+ j\mod \nbworkers}\}_{j=0}^\str$ on \worker\ $i\in \{0,\dots,\nbworkers-1\}$ results in an $(n,k,z) = (\nbworkers,\nbworkers-\str,1)$ sparse secret sharing with information leakage equal to $\dfrac{(\str+1)(\leakage{1}+\leakage{2})}{\nbworkers}$ and average storage cost (number of bits stored per node) bounded from above as
\begin{equation}\label{eq:storage_cost}
    \scalemath{0.9}{\mathbb{E}[C_\text{sparse}] = (2\str+2)(1-\savg)\dfrac{\m \n}{\nbworkers}\left(\left\lceil\log_2(q)\right\rceil + \left\lceil\log_2\left(\dfrac{\m\n}{\nbworkers}\right)\right\rceil\right)}.
\end{equation}
The average storage cost of this scheme is smaller than that of a non-sparse secret sharing scheme for an average sparsity satisfying
\begin{equation}
    \savg > 1 - \dfrac{1}{2\str}\cdot\dfrac{\nbworkers}{\nbworkers-\str}\dfrac{\log_2(q)}{\log_2(q) + \lceil\log_2(\nicefrac{\m\n}{\nbworkers})\rceil}. 
\end{equation}
\end{theorem}

\begin{proof}

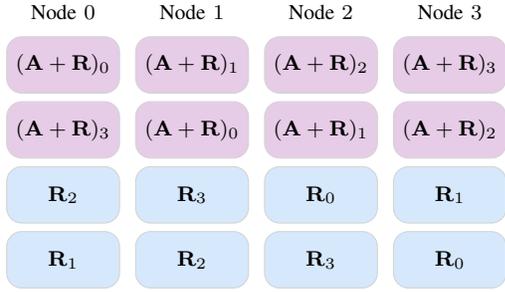
\begin{figure}[t]
    \begin{adjustbox}{center=.95\textwidth}
    	\begin{tikzpicture}[>=stealth', auto,
		triangle/.style = {fill=white, regular polygon, regular polygon sides=3 }]
		\definecolor{blue}{rgb}{0.19, 0.55, 0.91}
		\definecolor{plum}{rgb}{0.8, 0.6, 0.8}
		\def\blue{\color{blue}}
		\def\ogreen{\color{ogreen}}
		\def\orange{\color{orange}}
		\def\nred{\color{plum}}
		\definecolor{lightgray}{rgb}{0.83, 0.83, 0.83}
		\def\mx{0.4}
		\def\my{1}
		\tikzset{auto shift/.style={auto=right,->,
				to path={ let \p1=(\tikztostart),\p2=(\tikztotarget),
					\n1={atan2(\y2-\y1,\x2-\x1)},\n2={\n1+180}
					in ($(\tikztostart.{\n1})!1mm!270:(\tikztotarget.{\n2})$) -- 
					($(\tikztotarget.{\n2})!1mm!90:(\tikztostart.{\n1})$) \tikztonodes}}}
		\tikzstyle{matrixa} = [rectangle, rounded corners = 2mm, draw = lightgray, minimum height = 2.2cm, minimum width = 1.5 cm, fill = blue!40]
		\tikzstyle{matrixr} = [rectangle, rounded corners = 2mm, draw = lightgray, minimum height = 2.2cm, minimum width = 1.5 cm, fill = blue!20]
		
		\tikzstyle{matrixa1} = [rectangle, rounded corners = 1mm, draw = lightgray, minimum height = 0.75 cm, minimum width = 1.5 cm, fill = blue!40]
		\tikzstyle{matrixr1} = [rectangle, rounded corners = 2mm, draw = lightgray, minimum height = 0.75cm, minimum width = 1.5 cm, fill = blue!20]
				
		\tikzstyle{matrixa2} = [rectangle, rounded corners = 1mm, draw = lightgray, minimum height = 0.75cm, minimum width = 1.5 cm, fill = blue!40]
		
		\tikzstyle{matrixb} = [rounded corners = 2mm, draw = lightgray, minimum height = 1.5cm, minimum width = 0.3 cm, fill = ogreen!30]
		
		\tikzstyle{tasks1} = [rounded corners = 2mm, draw = lightgray, minimum height = 0.75cm, minimum width = 1.5 cm, fill = plum!50]
		
		\tikzstyle{tasks2} = [rounded corners = 2mm, draw = lightgray, minimum height = 0.75cm, minimum width = 1.5 cm, fill = blue!20]
		
 		\tikzstyle{server} = [fill=black!10, rectangle, rounded corners=4mm, draw,minimum width=2em, minimum height=2.5em]
 		
 		\node[inner sep=0pt, font=\footnotesize] at (0,0) (w1){};
 		\node[below =-0.3 of w1] (ww'2){};
 		\node[inner sep=0pt, left = 2 of w1, font=\footnotesize] (w'1){};
 		\node[below =-0.3 of w'1] (ww'1){};
 		\node[inner sep=0pt, right = 3cm of w1, font=\footnotesize] (w2){};
 		\node[below =-0.3 of w2] (ww'N){};
 		\node[inner sep=0pt, right = 3cm of w2, font=\footnotesize] (w3){};
 		\node[below =-0.3 of w3] (ww1){};
 		\node[inner sep=0pt, right = 2cm of w3, font=\footnotesize] (w4){};
 		\node[below =-0.3 of w4] (ww2){};
 		\node[inner sep=0pt, right = 3cm of w4, font=\footnotesize] (w5){};
 		\node[below =-0.3 of w5] (wwN){}; 
 		\node[tasks1, below left = 0 and 0 of ww'1] (A1+R1) {\footnotesize$(\bfA+\bfR)_{0}$};
 		\node[tasks1, right = 0.2 of A1+R1] (A2+R2) {\footnotesize$(\bfA+\bfR)_{1}$};
 		\node[tasks1, right = 0.2 of A2+R2] (A3+R3) {\footnotesize$(\bfA+\bfR)_{2}$};
 		\node[tasks1, right = 0.2 of A3+R3] (A4+R4) {\footnotesize$(\bfA+\bfR)_{3}$};
 		\node[tasks1, below = 0.1 of A1+R1] (A'1+R'1) {\footnotesize$(\bfA+\bfR)_{3}$};
 		\node[tasks1, below = 0.1 of A2+R2] (A'2+R'2) {\footnotesize$(\bfA+\bfR)_{0}$};
 		\node[tasks1, below = 0.1 of A3+R3] (A'3+R'3) {\footnotesize$(\bfA+\bfR)_{1}$};
 		\node[tasks1, below = 0.1 of A4+R4] (A'4+R'4) {\footnotesize$(\bfA+\bfR)_{2}$};
 		\node[tasks2, below = 0.1 of A'1+R'1] (R3) {\footnotesize$\bfR_{2}$};
 		\node[tasks2, below = 0.1 of A'2+R'2] (R4) {\footnotesize$\bfR_{3}$};
		\node[tasks2, below = 0.1 of A'3+R'3] (R1) {\footnotesize$\bfR_{0}$};
		\node[tasks2, below = 0.1 of A'4+R'4] (R2) {\footnotesize$\bfR_{1}$};
		
		\node[tasks2, below = 0.1 of R3] (R'2) {\footnotesize$\bfR_{1}$};
 		\node[tasks2, below = 0.1 of R4] (R'3) {\footnotesize$\bfR_{2}$};
		\node[tasks2, below = 0.1 of R1] (R'4) {\footnotesize$\bfR_{3}$};
		\node[tasks2, below = 0.1 of R2] (R'1) {\footnotesize$\bfR_{0}$};
		
 		\node[above = 0.1 of A1+R1] (W'1) {\footnotesize {\Worker} $0$};
 		\node[above = 0.1 of A2+R2] (W'2) {\footnotesize {\Worker} $1$};
 		\node[above = .1 of A3+R3] (W'3) {\footnotesize {\Worker} ${2}$};
 		\node[above = .1 of A4+R4] (W'4) {\footnotesize {\Worker} ${3}$};

		\end{tikzpicture}
    \end{adjustbox}
    \caption{An illustration of the proposed sparse secret sharing scheme for $\nbworkers=4$ {\workers} and $\str = 1$. Observe that the storage of any $3$ {\workers} is enough to reconstruct the secret $\bfA$.}
    \label{fig:fractional_repetition_sparse}
\end{figure}

We start by proving that the storage of any $\nbworkers-\str$ nodes is enough to reconstruct the secret. Our share distribution can be seen as a fractional repetition scheme with cyclic shifts. More precisely, each sub-share is replicated $\str+1$ times and is given to $\str+1$ distinct nodes. Therefore, any collection of $\nbworkers-\str$ workers stores all the sub-shares. Having access to all the sub-shares, we can obtain the shares $\bfR$ and $\bfA+\bfR$ and hence obtain $\bfA = \bfA+\bfR-\bfR$.

Since we focus on the case of $z=1$, we only need to quantify the per-node leakage. Note that each node obtains $\str+1$ sub-shares of the form $\bfR_i$ and $\str+1$ sub-shares of the form $(\bfA+\bfR)_i$. Since the entries of $\rvA$ and $\rvR$ are assumed to be i.i.d., the sub-shares are independent and cannot be combined to obtain more information about $\rvA$. Hence, the per-node information leakage is the sum of the information leaked from each sub-share and is equal to $\dfrac{(\str+1)(\leakage{1}+\leakage{2})}{\nbworkers}$.

To obtain an upper bound on the average storage cost of our sparse secret sharing scheme, we assume a naive compression scheme\footnote{There exist better compression algorithms that are limited to special criteria of the sparse matrices. However, we consider the general case as a worst-case scenario.}. A sparse matrix is compressed by storing only the value of the non-zero entries and their corresponding indices. The average storage cost of each sub-share is then the average number of non-zero entries multiplied by $\left\lceil\log_2(q)\right\rceil$ bits to store their values plus $\left\lceil\log_2(\m\n/ \nbworkers)\right\rceil$ bits to store their index values. Note that here the index ranges from zero to $\m\n / \nbworkers$ since each sub-share is a matrix of dimension $\m / \nbworkers \times \n$. Thus, since each node stores $2\str+2$ sub-shares, the average storage cost of our sparse secret sharing scheme $\mathbb{E}[C_\text{sparse}]$ is bounded from above as in \cref{eq:storage_cost}.

On the other hand, a classical $(n,k,z)=(\nbworkers, \nbworkers-\str,1)$ secret sharing scheme has a storage cost $C_\text{classical} = \frac{\m\n}{\nbworkers - \str-1}\lceil\log_2(q)\rceil$. As a result, $\mathbb{E}[C_\text{sparse}] < C_\text{classical}$ holds when $\savg$ satisfies the following
\begin{equation}
    s_\text{avg} > 1 - \dfrac{1}{2\str}\cdot\dfrac{\nbworkers}{\nbworkers-\str}\cdot\dfrac{\log_2(q)}{\log_2(q) + \lceil\log_2(\nicefrac{\m\n}{\nbworkers})\rceil}. 
\end{equation}
\cref{tab:param} shows the minimal values of $\savg$ for which our scheme is beneficial for specific parameters and $\m\n=10^{20}$. The relative leakage $\bar{\varepsilon}$ shown in \cref{tab:param} is equal to the information leakage normalized by $\entropy(\rvA)$.%
\end{proof}

\begin{table}[t]
    \centering
    \small
    \caption{Minimum average sparsity levels for the shares and their underlying minimum relative leakage for $\m\n = 10^{20}$.}
    \begin{tabular}{ccccccc}
    \FL
    $\log_2(q)$ & $s$ & $\str$ & $N$ & $s_\text{avg}$ & $\bar{\varepsilon}$ \ML\midrule \\[-2ex]
    $32$ & $0.95$ & $2$ & $60$ & $0.9396$ & $0.0135$ \\ \midrule
    $32$ & $0.95$ & $4$ & $60$ & $0.9625$ & $0.0502$  \\ \midrule
    $32$ & $0.99$ & $4$ & $60$ & $0.9625$ & $0.0206$  \\ \midrule
    $32$ & $0.99$ & $5$ & $100$ & $0.9692$ & $0.0206$ \\ \midrule
    $20$ & $0.99$ & $5$ & $100$ & $0.9778$ & $0.0176$ \\\bottomrule
    \end{tabular}
    \label{tab:param}
\end{table}

\section{Conclusion and Future Directions}\label{sec:conc}

We investigated a sparsity-preserving secret sharing scheme for private distributed storage of machine learning data. The goal is to create sparse shares that can be efficiently stored at the expense of relaxing the privacy requirement. We focused on secret sharing schemes that take as input a private matrix $\bfA$, generate a random matrix $\bfR$ of the same size and output two shares $\bfR$ and $\bfA+\bfR$. Assuming the non-zero entries of the private matrix $\bfA$ being uniformly distributed and under some mild restrictions on the conditional PMF of $\rvR$, we showed that there is a trade-off between the sparsity of the shares and the amount of information they leak about $\bfA$. For desired sparsity guarantees, we provided a convex optimization problem that gives the minimum leakage and solved this optimization for special distributions of $\rvR$. We observed that the optimal leakage can be achieved when $\rvR$ and $\rvA+\rvR$ have same sparsity levels $\sar=\sar$. Proving this observation analytically is left for future work. 

Furthermore, we showed how to divide the shares into sub-shares and distribute them to $\nbworkers$ nodes using fractional repetition codes to create an $(n,k,z)$ sparse secret sharing schemes with $k\in \{\frac{\nbworkers}{2}+1,\dots,\nbworkers\}$ and $z=1$. Compared to classical secret sharing, for large values of the desired sparsity, our scheme saves on storage at the expense of relaxing the privacy requirement. Improving on the rate of sparsity-preserving secret sharing schemes and designing more efficient suitable coding schemes for straggler tolerance are left for future investigation.

\balance
\bibliographystyle{ieeetr}
\bibliography{IEEEabrv,main}

\end{document}